\documentclass[letterpaper, 10pt, conference]{ieeeconf}      

\IEEEoverridecommandlockouts                              

\usepackage{graphics} 
\usepackage{epsfig} 
\usepackage{mathptmx} 
\usepackage{times} 
\usepackage{amsmath} 
\usepackage{amssymb}  
\usepackage{cite}
\usepackage{graphicx} 
\usepackage{tabularx}
\usepackage{hyperref}
\usepackage{xcolor}
\usepackage{dsfont}
\usepackage{multicol}
\usepackage{subcaption}

\newtheorem{thm}{Theorem}
\newtheorem{lem}{Lemma}
\newtheorem{prop}{Proposition}
\newtheorem{rem}{Remark}
\newtheorem{defi}{Definition}
\newtheorem{assum}{Assumption}

\DeclareMathOperator*{\argmin}{arg\,min}
\DeclareMathOperator*{\poa}{PoA}

\title{\LARGE \bf On the impact of coordinated fleets size on traffic  efficiency}

\author{Tommaso Toso, \IEEEmembership{Student Member, IEEE}, Francesca Parise, \IEEEmembership{Member, IEEE}, Paolo~Frasca,~\IEEEmembership{Senior~Member,~IEEE},\\
Alain Y. Kibangou, \IEEEmembership{Member, IEEE}
\thanks{This work has been partly supported by the French National Research Agency in the framework of the "Investissements d’avenir” program ANR-15-IDEX-02 and the LabEx PERSYVAL ANR-11-LABX-0025-01. T.~Toso thanks the support of an outgoing mobility grant from Univ. Grenoble Alpes.}
\thanks{Tommaso Toso, Paolo Frasca and Alain Kibangou are with Univ.\ Grenoble Alpes, CNRS, Inria, Grenoble INP, GIPSA-lab, 38000 Grenoble, France (e-mail:firstname.lastname@gipsa-lab.grenoble-inp.fr).
Alain Kibangou is also with Univ. of Johannesburg (Auckland Park Campus), Johannesburg 2006, South Africa. Francesca Parise is with the School of Electrical and Computer Engineering,
Cornell University, Ithaca, NY, USA (e-mail: fp264@cornell.edu).
} }

\begin{document}

\maketitle


\begin{abstract}

We investigate a traffic assignment problem on a transportation network, considering both the demands of  individual drivers and of a large fleet controlled by a central operator (minimizing the fleet's average travel time). We formulate this problem as a two-player convex game and we study how the size of the coordinated fleet, measured in terms of share of the total demand, influences the Price of Anarchy (PoA). We show that, for two-terminal networks, there are cases in which the fleet must reach a minimum share before actually affecting the PoA, which otherwise remains unchanged. Moreover, for parallel networks we prove that, under suitable assumptions, the PoA is monotonically non-increasing in the fleet share. 
\end{abstract}



\section{Introduction}

Traffic assignment problems typically assume that traffic demand consists of drivers exhibiting selfish behavior to minimize travel time. However, technological advancements  have introduced new services (ride-hailing services, navigation apps) that, due to their widespread adoption, can influence the behavior of a substantial portion of drivers, potentially leading to a paradigm shift. Specifically, the providers of such services may leverage their position to minimize overall fleet metrics, such as total travel time, rather than optimizing individual user experiences. This approach, while potentially disadvantaging some users, aims to attract and retain users by providing lower travel times on average. In the following,  we refer to groups of vehicles controlled by a central operator aiming to minimize the fleet's average travel time as \textit{coordinated fleets}. 
This work aims to study the impact that the presence of a coordinated fleet has on traffic efficiency in terms of the price of anarchy (PoA).

\subsubsection*{Contribution}

We formulate the problem as a two-player game, with one player  associated with the individual users and the other with the coordinated fleet. We study this game by using a well-known reformulation in terms of solution to a Variational Inequality (VI) (see \cite{vis,scutari}). Specifically, we establish conditions ensuring that the operator of the VI associated to our game is strongly monotone. On the one hand, strong monotonicity ensures equilibrium  uniqueness. On the other hand, this property allows for providing meaningful insights about the relationship between traffic efficiency and the share of the coordinated fleet in two-terminal networks. Using the  Price of Anarchy (PoA) as a measure of traffic efficiency \cite{rough}, we prove that the unique equilibrium and the PoA are Lipschitz continuous functions of the fleet share. Additionally, we derive sufficient conditions for the existence of a minimum share below which the  presence of a coordinated fleet has no effect on traffic efficiency. 
Finally, for parallel networks, we show that under suitable assumptions, the PoA is monotonically non-increasing in the fleet share.

\subsubsection*{Related work}
The multi-class traffic assignment problem was initially defined in \cite{dafermos}. Coordination among users of the same class was introduced in \cite{harker}, where sufficient conditions for  equilibrium existence and uniqueness are established, and then extended to more general settings in \cite{boulogne}. 
The impact on efficiency of coordinated classes was first considered in \cite{yang} for a three-class problem with \textit{individual users}, aiming at reducing individual travel time, \textit{a coordinated  fleet}, aiming at reducing overall fleet travel time; and a \textit{system-optimal fleet}, aiming at reducing the system's average travel time.
Numerical simulations in \cite{yang} show that  sufficiently large coordinated and system-optimal fleets can lead to system optimality even in the presence of individual users. 

Two-class problems are considered in \cite{sharon,chen,zhang,nilsson, batti}.
Specifically, \cite{sharon,chen,zhang} consider a two-class problem, with individual users and a system-optimal fleet:  \cite{sharon} and \cite{chen} derived methods to compute the minimum share of system-optimal users necessary to induce system optimality, while \cite{zhang} studied the trade-off between the magnitude of the improvement and the cost of deployment for the network manager.

More closely related to our contribution,  \cite{nilsson} and \cite{batti}  both consider a two-class problem with individual users and one coordinated fleet. 
In \cite{nilsson}, the authors derive sufficient conditions for equilibrium existence and uniqueness weaker than those used in our paper. These conditions do not guarantee strong monotonicity, which is essential for our analysis. Their work also introduces two algorithms for computing the equilibrium and a control scheme for achieving equilibrium in a dynamic framework.


In \cite{batti}, the authors also examine the impact of coordinated fleets on traffic efficiency, demonstrating that these fleets can reduce efficiency in networks with multiple origin-destination pairs. They analyze the minimum fleet size required for system optimality and the maximum size where user equilibrium still holds, providing mathematical programs to compute these values. Additionally, they offer analytical results on the threshold effect of fleet size on efficiency, but only for parallel networks. In contrast, our work provides conditions under which a minimum fleet size impacts the PoA for general networks with a single origin-destination pair and derives results on the monotonicity of the PoA in parallel networks.



\subsubsection*{Paper organization}
The model and the main concepts are defined in Section~\ref{s2}. Strong monotonicity, existence and uniqueness conditions are given in Section~\ref{s3}. Section~\ref{s4} discusses the effect of a coordinated fleet  on traffic efficiency as a function of the  fleet size. Section~\ref{s5} contains numerical studies and provides interesting insights about extending this work to more general settings. Section~\ref{s6} contains concluding remarks and future perspectives. Due to space limitations, proofs of some auxiliary results are omitted and can be found in \cite{arxiv}.

\section{A two-class routing game}\label{s2}

The transportation network is modeled as a directed graph $\mathcal{G}=(\mathcal{N},\mathcal{L})$, with node set $\mathcal{N}$ and link set $\mathcal{L}$, with links representing roads of the network and nodes representing junctions between them. 
Let $\mathcal{O}\subseteq\mathcal{N}$, called \textit{origins}, be the subset of nodes from which exogenous traffic demands can access the network. Analogously, let $\mathcal{D}\subseteq\mathcal{N}$, called \textit{destinations}, be the subset of nodes through which traffic can exit the network.
Define the set of \textit{origin-destination pairs} (OD pairs) $\mathcal{K}=\{(o,d)\,|\,o\in\mathcal{O},d\in\mathcal{D}\}$. Let $\mathcal{P}_k$ denote the set of paths associated with OD pair $k$ and let $\mathcal{P}:=\cup_{k\in\mathcal{K}}\mathcal{P}_k$. 
Let $N$, $L$, $K$, $P_k$ and $P$ be the cardinalities of $\mathcal{N}$, $\mathcal{L}$,  $\mathcal{K}$, $\mathcal{P}_k$ and $\mathcal{P}$, respectively. 
Let $A$ be the \textit{link-path incidence matrix} defined as
\begin{equation}
    A_{lp}=\begin{cases}
        1, & l\in p\\
        0, & l\notin p
    \end{cases}.
\end{equation}

Suppose now that $\mathcal{G}$ supports two classes of demand, namely class $S$ and class $C$. Class $S$ consists of selfish individual users, whereas class $C$ consists of a coordinated fleet. Let $D^i$ be the total demand of class $i,\ i=S,\,C$. Each OD pair $k$ is associated with a fraction $D^i_k>0,\ i=S,\,C$, of the total demand, i.e, $\sum_{k\in\mathcal{K}}D^i_k=D^i$. Let $D:=D^S+D^C$ be the total demand. For each class $i\in\{S,\,C\}$, we define the \textit{flow vector of class $i$} $z^i\in\mathbb{R}_{\geq0}^P$ representing the traffic assignment of traffic demand $D^i$ over the network paths. The \textit{set of feasible flows of class} $i$ is 
\begin{equation*}
    \mathcal{Z}^i:=\biggl\{z^i\in\mathbb{R}_{\geq0}^P\,:\,\sum_{p\in\mathcal{P}_k}z_p^i=D_k^i,\ \forall k\in\mathcal{K}\biggl\}
\end{equation*}
and let $\mathcal{Z}=\mathcal{Z}^S\times\mathcal{Z}^C$.
Each flow vector $z^i$ is associated with the \textit{load vector of class $i$} ($f^i:=Az^i,\ i=S,\,C$) representing the load of each link of the network for class $i$. Then, the set of feasible loads of class $i$ is 
\begin{equation*}
    \mathcal{F}^i:=\{f^i\in\mathbb{R}_{\geq0}^L\,:\,f^i=Az^i,\text{ for some } z^i\in\mathcal{Z}^i\}
\end{equation*}
and let $\mathcal{F}=\mathcal{F}^S\times\mathcal{F}^C$.
Let the \textit{flow vector} $ z:=(z^S,z^C)$ and the \textit{load vector }$f:=(f^S,f^C)$ be the concatenations of the flow and load vectors of the two classes and let $Z:=z^S+z^C$, $F:=f^S+f^C$ be the \textit{aggregate flow} and \textit{load} vectors, respectively. 
The assignment of the two classes of vehicles is determined by the link delay functions.
\begin{defi}[Delay functions]
        For every $l\in\mathcal{L}$, the delay  $d_l:\mathbb{R}_{\geq0}\rightarrow\mathbb{R}_{\geq0}$ of link $l$ is a non-negative, strictly increasing and $C^2([0,+\infty))$ function with $d_l'(0)>0$,
        depending on the aggregate load $F_l$ on link $l$ only. Moreover, for every $p\in\mathcal{P}$, the function $d_p:\mathbb{R}^L_{\geq0}\rightarrow\mathbb{R}_{\geq0}$ is the \textit{delay of path} $p$ and equals the sum of the delays of the links included in $p$:
    \begin{equation}
        d_p(F)=\sum_{l\in\mathcal{L}}A_{lp}d_l(F_l).
    \end{equation}
    \label{def:df}
\end{defi} 

The fact that link delays depend only on the aggregate load means that the two classes of vehicles affect the link delays in the same way. 

We are interested in characterizing the equilibrium loads of the traffic assignment problem emerging from the interaction of the vehicle classes $S$ and $C$. To do this, we reformulate the problem as a \textit{two-player game}, by associating each class to a strategic player.  The strategy of each player corresponds to the load vector $f^i$ with strategy set $\mathcal{F}^i,\ i=S,\,C$, respectively. The cost functions that player $S$ and player $C$ have to minimize in order to attain the goals of the traffic assignment problem are the following:
\begin{equation}\label{optaS}
U^S(f):=\sum_{l\in\mathcal{L}}\, \int_0^{{f_l^S}}d_l(r+{f_l^C})dr,
\end{equation}
\begin{equation}
U^C(f):= \sum_{l\in\mathcal{L}}\,f_l^C\cdot d_l(F_l).
\label{opta}
\end{equation}
In deriving the cost function for player $S$ we used a well-know reformulation of the Wardrop equilibrium of strategic agents in class $S$ as an optimization problem (with potential function as in \eqref{optaS}),\cite[Chapter 18]{rough}. The cost of the player $C$ instead is the total travel time of vehicles in class $C$. 
 
\begin{defi}[Equilibria]
   An \textit{equilibrium load} is a load vector $f^*=({f^S}^*,{f^C}^*)$ such that 
   \begin{equation}
   \begin{aligned}
       {f^S}^*:=\argmin_{f^S\in\mathcal{F}^S}\ U^S(f^S,{f^C}^*),\ \ {f^C}^*:=\argmin_{f^C\in\mathcal{F}^S}\ U^C({f^S}^*,f^C).
       \end{aligned}
       \label{el}
   \end{equation}
  All the feasible flows $z^*=({z^S}^*,{z^C}^*)$ such that $f^{i^*}=Az^{i^*},\ i=S,\,C$ are called \textit{equilibrium flows}.
\end{defi}

Note that 
\begin{equation*}
    f_l^C\cdot d_l(F_l)=\int_0^{f_l^C}(d_l({f_l^S}+r)+r\cdot d_l'({f_l^S}+r))\,dr.
\end{equation*}
Hence \eqref{opta} can be rewritten as  
\begin{equation}
     U^C(f)=\sum_{l\in\mathcal{L}}\,\int_0^{f_l^C}(d_l({f_l^S}+r)+r\cdot d_l'({f_l^S}+r))\,dr.
    \label{opta2}
\end{equation}
The  functions inside the integral in \eqref{opta2}, that is, 
\begin{equation} 
m_l(f_l):=d_l(F_l)+f_l^C\cdot d_l'(F_l)
    \label{md}
\end{equation}
are known as \textit{marginal delay functions} \cite[Chapter 18]{rough}.

We prove that under appropriate assumptions on the marginal delays $m_l(f_l)$, the  game in \eqref{el} is convex.
\begin{lem} \label{lem1}
$U^S(f)$ is convex in $f^S$ for any $f^C$. Moreover, if 
\begin{equation}
 \frac{\partial m_l(f_l)}{\partial f_l^C}>0,\quad  \forall f_l^S,\,f_l^C\geq0,\ \forall l\in\mathcal{L},
 \label{convC}
\end{equation}
   then $U^C(f)$ is convex in $f^C$ for any $f^S$.  
\end{lem}
\begin{rem}
    The convexity of the cost functions \eqref{optaS} and \eqref{opta} implies that any equilibrium flow $z^*$ must satisfy the following Wardrop  conditions: 
    \begin{equation}
            {z^S_p}^*>0 \Rightarrow \sum_{l\in\mathcal{L}}A_{lp}d_l(F_l^*)\leq\sum_{l\in\mathcal{L}}A_{lr}d_l(F_l^*),\ \forall r\in\mathcal{P},
            \label{w1}
    \end{equation} 
    \begin{equation}
            {z^C_p}^*>0 \Rightarrow \sum_{l\in\mathcal{L}}A_{lp}m_l(f_l^*)\leq\sum_{l\in\mathcal{L}}A_{lr}m_l(f_l^*),\ \forall r\in\mathcal{P}.
            \label{w2}
    \end{equation} 
    In words, at equilibrium, each vehicle in class $S$ uses a path among those of shortest delay, whereas each vehicle in class $C$ uses a path among those of shortest marginal delay. Conditions \eqref{w1} and \eqref{w2} will be of key importance when proving the results  in Section \ref{s4}.
\end{rem}

\section{Variational Inequality formulation and strong monotonicity}\label{s3}
Under condition  \eqref{convC}, the two-class routing game is convex and is equivalent to the following variational inequality (VI) \cite[Proposition 1.4.2]{vis}:
\begin{equation}
    (\phi-f^*)^\top H(f^*)\geq0,\quad \forall\phi\in \mathcal{F},
    \label{vi}
\end{equation}
where 
\begin{equation}
    H(f)=\left((d_l(F_l))_{l\in\mathcal{L}},(m_l(f_l))_{l\in\mathcal{L}}\right),
    \label{ope}
\end{equation}
that is, equilibria of the two-class routing game correspond to  solutions of \eqref{vi}. 

The main result of this section consists in providing sufficient conditions for the operator $H$ to be strongly monotone on $\Omega:=[0,D]^{2L}\supset\mathcal{F}$, that is, for guaranteeing that 
\begin{equation}
     \exists c>0\,:\,(H(x)-H(y))^\top (x-y)\geq c||x-y||^2,\ \forall x,y\in\Omega.
     \label{def:sm}
 \end{equation}
The strong monotonicity of $H$ not only ensures the uniqueness of the solution of \eqref{vi} \cite[Theorem 2.3.3]{vis}, but also allows us to  assess the impact of the fleet size onto traffic efficiency, as we shall demonstrate in the next section.

\begin{prop}
The operator $H$ in \eqref{ope} is strongly monotone on $\Omega$ if \eqref{convC} holds and 
\begin{equation}
 d'_l(F_l)>\frac{1}{4}\frac{\partial m_l(f_l)}{\partial f_l^C},\quad  \forall f_l^S,\,f_l^C\geq0,\ \forall l\in\mathcal{L}.
\label{unique}
\end{equation}
\label{p1}
\end{prop}
   The strong monotonicity of $H$ ensures the uniqueness of the solution of \eqref{vi}, that is, of the equilibrium load $f^*$. In \cite{nilsson}, weaker conditions  similar to \eqref{unique} were derived to ensure the uniqueness of the equilibrium load. Our slightly stronger conditions are needed to guarantee that $H$ is strongly monotone and that thus the following assumption holds. 
\begin{assum}\label{as:sm}
  Suppose that the operator $H$ in \eqref{ope} is Lipschitz and strongly monotone in $\Omega=[0,D]^{2L}$. 
\end{assum}
Again, sufficient conditions for strong monotonicity to hold are given in Proposition \ref{p1}, whereas Lipschitz continuity  follows  from the smoothness of delay and marginal delay functions (defined on a compact set).

\begin{rem}
    A class of delay functions that satisfy conditions \eqref{convC} and \eqref{unique}, ensuring strong monotonicity of \eqref{ope}, includes polynomial functions of degree up to 3 with non-negative coefficients and strictly positive derivatives on \([0, +\infty)\) (see \cite{nilsson} for similar examples). This indicates that assuming strong monotonicity is not overly restrictive, as it applies to a significant class of delay functions.
\end{rem}

\section{Price of Anarchy}\label{s4}
The \textit{total delay} experienced by all the vehicles traveling across the network is defined as 
\begin{equation}
    T(f):=\sum_{l\in\mathcal{L}}F_l\cdot d_l(F_l).
    \label{ttt}
\end{equation}
A feasible load  minimizing $T(f)$ is called an \textit{optimal load} (denoted by $F^\omega\in\mathcal{F}$). Then, the \textit{Price of Anarchy}  is defined as the ratio between  the total delay attained at the (unique under Assumption \ref{as:sm}) equilibrium $f^*$ and the minimum total delay:
\begin{equation}
    \poa :=\frac{\sum_{l\in\mathcal{L}}F_l^*\cdot d_l(F_l^*)}{\sum_{l\in\mathcal{L}}F_l^\omega\cdot d_l(F_l^\omega)}\ge1.
    \label{poa}
\end{equation}
We aim  to study how the size of the coordinated fleet affects the $\mathrm{PoA}$. From now on, we  focus our attention on \textit{two-terminal networks}.
\begin{assum}\label{a:od}
The network has a single OD pair. Let  $D^S=(1-\alpha)D$ and $D^C=\alpha D$ represent the demand of class $S$ and $C$ entering the network from its unique origin, where $\alpha$  is the share of class $C$, which we refer to as the \textit{fleet share}. 
\end{assum}


We present three main results in this section. First, we show that the equilibrium load and the PoA are Lipschitz continuous functions of the fleet share $\alpha$. Second, we establish a sufficient condition for a minimum fleet share below which the coordinated fleet has no effect on the PoA. Finally, we prove that  under suitable assumptions the PoA is a non-increasing function of $\alpha$ for parallel networks. To clarify their dependence on $\alpha$, we denote the feasible set by $\mathcal{F}(\alpha)$, the equilibrium load by $f^*(\alpha)$, and the PoA, delay, and marginal delay functions at equilibrium by $\poa(\alpha)$, $d_l(\alpha)$, and $m_l(\alpha)$, respectively.


\subsection{Lipschitz continuity}
\begin{prop}
 Let Assumptions \ref{as:sm} and \ref{a:od} hold. The equilibrium load $f^*(\alpha)$ is Lipschitz continuous in $\alpha$, i.e., there exists $k>0$ such that  
 \begin{equation}
     \forall \alpha_1,\alpha_2,\quad ||f^*(\alpha_2)-f^*(\alpha_1)||\leq k|\alpha_2-\alpha_1|.
     \label{lc}
 \end{equation}
 \label{p2}
\end{prop}
Since the PoA is Lipschitz continuous in the equilibrium load and loads are defined on a bounded set, the PoA is also a Lipschitz continuous function of $\alpha$.

\subsection{Critical fleet share}
A first question that one may ask is if introducing a coordinated fleet always reduces the PoA. In this section, we show that this is not necessarily the case. We derive a sufficient condition for the existence of a positive minimum fleet size required to affect the PoA.

\begin{thm}[Critical fleet size]
    \label{t1}
    Let 
$\mathcal{P}^i(z(\alpha))$
indicate the set of paths used by class $i$ at the equilibrium flow $z(\alpha)$, $i=C,S$, respectively.
     Let Assumptions \ref{as:sm} and \ref{a:od} hold. Suppose that $\exists\Tilde{\alpha}\in(0,1)$ such that $f^*(\Tilde{\alpha})$ admits an equilibrium flow $z^*(\Tilde{\alpha})$ such that $\mathcal{P}^C(z^*(\Tilde{\alpha}))\subseteq\mathcal{P}^S(z^*(\Tilde{\alpha}))$\footnote{
Verifying this hypothesis can be difficult in practice. The derivation of alternative assumptions in left as future work.}. Then, 
     \begin{equation}
        f^*(\alpha)=\left({f^S}^*(\Tilde{\alpha})+\frac{\Tilde{\alpha}-\alpha}{\Tilde{\alpha}}{f^C}^*(\Tilde{\alpha}),\frac{\alpha}{\Tilde{\alpha}}{f^C}^*(\Tilde{\alpha})\right),
        \label{cand}
    \end{equation}
    and $\mathrm{PoA}(\alpha)=\mathrm{PoA}(0),\ \forall \alpha\in[0,\Tilde{\alpha}]$. 
\end{thm}
\begin{proof}
   Consider the equilibrium load $f^*(\Tilde{\alpha})$ and the associated  equilibrium flow $z^*(\Tilde{\alpha})$. Clearly,
   \begin{equation}
   \begin{aligned}
       {z_p^S}^*(\Tilde{\alpha})>0 \Rightarrow \sum_{l\in\mathcal{L}}A_{lp}d_l(\Tilde{\alpha})&\leq\sum_{l\in\mathcal{L}}A_{l\gamma}d_l(\Tilde{\alpha}),\ \forall\gamma\in\mathcal{P}, \\
       {z_p^C}^*(\Tilde{\alpha})>0 \Rightarrow \sum_{l\in\mathcal{L}}A_{lp}m_l(\Tilde{\alpha})&\leq\sum_{l\in\mathcal{L}}A_{l\gamma}m_l(\Tilde{\alpha}),\ \forall\gamma\in\mathcal{P}.
       \end{aligned}
       \label{morec1}
   \end{equation}
     Consider the following feasible flow (obtained by moving part of the flow from C to S):
     \begin{equation}
        z^*(\alpha)=\left({z^S}^*(\Tilde{\alpha})+\frac{\Tilde{\alpha}-\alpha}{\Tilde{\alpha}}{z^C}^*(\Tilde{\alpha}),\,\frac{\alpha}{\Tilde{\alpha}}{z^C}^*(\Tilde{\alpha})\right).
        \label{cande}
    \end{equation}
     We show that $z^*(\alpha)$ is an equilibrium flow when the fleet share is $\alpha$, i.e., 
    \begin{equation}
    \begin{aligned}
       {z_p^S}^*(\alpha)>0 \Rightarrow \sum_{l\in\mathcal{L}}A_{lp}d_l(\alpha)&\leq\sum_{l\in\mathcal{L}}A_{l\gamma}d_l(\alpha),\ \forall\gamma\in\mathcal{P}, \\
       {z_p^C}^*(\alpha)>0 \Rightarrow  \sum_{l\in\mathcal{L}}A_{lp}m_l(\alpha)&\leq\sum_{l\in\mathcal{L}}A_{l\gamma}m_l(\alpha),\ \forall\gamma\in\mathcal{P}.
       \end{aligned}
       \label{morec2}
   \end{equation}
    We  prove each of the above conditions. The  first easily follows after noticing that i) $z^*(\alpha)$ and $z^*(\Tilde{\alpha})$ induce the same aggregate load, i.e., $F_l^*(\alpha)=F_l^*(\Tilde{\alpha})$, so none of the path delays has changed, and ii) the set of paths used by vehicles in class $S$ is the same, i.e., $\mathcal{P}^S(z^*(\Tilde{\alpha}))=\mathcal{P}^S(z^*(\alpha))$ (since $\mathcal{P}^C(z^*(\Tilde{\alpha}))\subseteq \mathcal{P}^S(z^*(\Tilde{\alpha}))$ ). Hence, the first inequality in \eqref{morec1} ensures that all vehicles in class $S$ still use shortest delay paths. As for the second condition, similarly, one has to prove that vehicles in class $C$ still use shortest marginal delay paths. Because of the expression of \eqref{cande}, one can observe that:
    \begin{itemize}
        \item $\mathcal{P}^C(z^*(\alpha))=\mathcal{P}^C(z^*(\Tilde{\alpha}))\subseteq \mathcal{P}^S(z^*(\Tilde{\alpha}))$;
        \item for every $p\in\mathcal{P}$, since the aggregate loads have not changed, the marginal delay is
        \begin{equation*}
            m_p(\alpha)=\sum_{l\in\mathcal{L}}A_{lp}\left(d_l(\Tilde{\alpha})+\frac{\alpha}{\Tilde{\alpha}}{f_l^{C}}^*(\Tilde{\alpha})d_l'(\Tilde{\alpha})\right),
        \end{equation*}
    \end{itemize}
    By multiplying the first inequality in \eqref{morec1} by $1-\alpha/\Tilde{\alpha}$, the second one by $\alpha/\Tilde{\alpha}$, then summing them, one obtains
        \begin{equation}
        \begin{aligned}
             m_p(\alpha)&=\sum_{l\in\mathcal{L}}A_{lp}\left(d_l(\Tilde{\alpha})+\frac{\alpha}{\Tilde{\alpha}}{f_l^{C}}^*(\Tilde{\alpha})d_l'(\Tilde{\alpha})\right)\leq\\ &\leq \sum_{l\in\mathcal{L}}A_{l\gamma}\left(d_l(\Tilde{\alpha})+\frac{\alpha}{\Tilde{\alpha}}{f_l^{C}}^*(\Tilde{\alpha})d_l'(\Tilde{\alpha})\right) =m_\gamma(\alpha),
        \end{aligned}
        \label{morec3}
        \end{equation}
    $\forall p\in\mathcal{P}^C(z^*(\alpha)),\ \forall \gamma\in\mathcal{P}$. Hence, every $\mathcal{P}^C(z^*(\alpha))$ is still a shortest marginal delay path.
    Therefore, $z^*(\alpha)$ is a equilibrium flow when the fleet share is equal to $\alpha$. The equilibrium load associated with $z^*(\alpha)$ is
    \begin{equation*}
       f^*(\alpha)=\left({f^S}^*(\Tilde{\alpha})+\frac{\Tilde{\alpha}-\alpha}{\Tilde{\alpha}}{f^C}^*(\Tilde{\alpha}),\frac{\alpha}{\Tilde{\alpha}}{f^C}^*(\Tilde{\alpha})\right),
    \end{equation*}
    which must correspond to the unique equilibrium of the problem.
     To conclude, notice that for all $\alpha\in[0,\Tilde{\alpha}]$ all links have the same aggregate load. Hence $\mathrm{PoA}(\alpha)=\mathrm{PoA}(0)$ for all $\alpha\in[0,\Tilde{\alpha}]$.
\end{proof}

\subsection{PoA monotonicity for Parallel Networks}

In this section, we study the dependence of the PoA on the fleet share $\alpha$ under the following assumptions.

\begin{assum}\label{a2}
    The network $\mathcal{G}$ is a parallel, thus it consists of an OD pair connected by finitely many links, all directed from the origin to the destination. Let $\alpha$ be the fleet share.
\end{assum}
\begin{assum}\label{a3}
    The delay function $d_l$ is convex, $\forall l\in\mathcal{L}$.
\end{assum}
The assumption of parallel networks simplifies the analysis as, in that case, the notion of link and path coincides. Let $\mathcal{L}^i(\alpha)$ indicate the set of links used by class $i$ at the equilibrium $f^*(\alpha)$, $i=S,C$.
The convexity of the delay functions instead ensures that 
$d'_l(F_l)$ is non-decreasing in $F_l$.
Note that this implies the following monotonicity property
\begin{equation}\label{mon:mu}
\bar F_l>F_l \ \textup{and} \ \bar f^C_l> f^C_l\quad  \Rightarrow \quad m_l(\bar f_l)>m_l(f_l).
\end{equation}



Let  $\theta(\alpha)$ and $\mu(\alpha)$ indicate the \textit{minimum delay} and the \textit{minimum marginal delay} realised at equilibrium when the fleet share is $\alpha$, respectively. Observe that, since links and paths coincide,  the equilibrium condition implies
\begin{equation*}
    l\in\mathcal{L}^S(\alpha)\Rightarrow d_l(F_l^*(\alpha))=\theta(\alpha),
\end{equation*}
\begin{equation*}
    l\in\mathcal{L}^C(\alpha)\Rightarrow m_l(f_l^*(\alpha))=\mu(\alpha).
\end{equation*}
\begin{prop}
    Let Assumptions \ref{as:sm}, \ref{a2} and \ref{a3} hold. Suppose there exists $\alpha_1,\alpha_2\in(0,1)$ such that $\alpha_1<\alpha_2$ and $\mathcal{L}^S(\alpha_1)=\mathcal{L}^S(\alpha_2)$ and $\mathcal{L}^C(\alpha_1)=\mathcal{L}^C(\alpha_2)$. Then, 
    \begin{enumerate}
        \item\label{i1} $\theta(\alpha_1)\geq\theta(\alpha_2)$;
        \item\label{i11} 
        $\mu(\alpha_1)\leq\mu(\alpha_2)$;
        \item\label{i2} ${f_l^S}^*(\alpha_1)\geq {f_l^S}^*(\alpha_2),\ \forall l\in \mathcal{L}$;
        \item\label{i3} ${f_l^C}^*(\alpha_1)\leq {f_l^C}^*(\alpha_2),\ \forall l\in \mathcal{L}$.
    \end{enumerate}
    \label{p4}
\end{prop}

\begin{proof}
Since $\mathcal{L}^i(\alpha_1)=\mathcal{L}^i(\alpha_2),\ i=S,\,C$, let us indicate both as $\mathcal{L}^i,\ i=S,\,C$ for convenience. Along with them, consider also the set $\mathcal{L}^{C\setminus S}:=\mathcal{L}^C\setminus(\mathcal{L}^S\cap\mathcal{L}^C)$, corresponding to the set of links used by class $C$ only. Notice that also this set remains constant in passing from $\alpha_1$ to $\alpha_2$. Also, since it is used by vehicles of class $C$ only,
\begin{equation}
        {f_l^C}^*(\alpha_i)=F_l^*(\alpha_i),\ \forall l\in\mathcal{L}^{C\setminus S},\ i=1,2.
        \label{coinc}
    \end{equation}

If  $\mathcal{L}^{C\setminus S}=\emptyset$ the conclusion follows from Theorem \ref{t1}. We next discuss the case in which    $\mathcal{L}^{C\setminus S}\neq \emptyset$. \\
1) By contradiction, suppose that $\theta(\alpha_1)<\theta(\alpha_2)$. This implies that the aggregate load increased on all links in $\mathcal{L}^S$, i.e., $F_l^*(\alpha_1)<F_l^*(\alpha_2),\ \forall l \in\mathcal{L}^S$. Now, since the demand of class $S$ decreased, there must exist a link $j\in\mathcal{L}^S$ such that the load of class $S$ on it decreased, i.e., ${f_j^S}^*(\alpha_1)>{f_j^S}^*(\alpha_2)$. The latter fact, combined with the increase of the aggregate loads on all link in $\mathcal{L}^S$, implies that the load of class $C$ on link $j$ increased, i.e., ${f_j^C}^*(\alpha_1)<{f_j^C}^*(\alpha_2)$. 
By \eqref{mon:mu}, the increase of both the aggregate load and the load of class $C$ on link $j$ implies that its marginal delay increased. Hence, $\mu(\alpha_1)<\mu(\alpha_2)$.
    
On the other hand, the increased load on links in $\mathcal{L}^S$ implies higher demand towards $\mathcal{L}^S$ and lower demand towards $\mathcal{L}^{C\setminus S}$. Then, there must be at least one link $e\in\mathcal{L}^{C\setminus S}$ whose aggregate load decreased, i.e., $F_e^*(\alpha_1)>F_e^*(\alpha_2)$. From \eqref{coinc}, this is equivalent to ${f_e^C}^*(\alpha_1)>{f_e^C}^*(\alpha_2)$, which implies that  $\mu(\alpha_1)>\mu(\alpha_2)$, contradicting what proved above. Therefore, $\theta(\alpha_1)\geq\theta(\alpha_2)$.
    \\
    2) From \ref{i1}), $\theta(\alpha_1)\geq\theta(\alpha_2)$, which implies that the aggregate load on none of the links in $\mathcal{L}^S$ can increase. Hence, the aggregate demand toward $\mathcal{L}^S$ cannot increase, which is equivalent to say that the aggregate demand toward $\mathcal{L}^{C\setminus S}$ cannot decrease. From \eqref{coinc}, it follows that the demand associated with class $C$ directed toward $\mathcal{L}^{C\setminus S}$ did not decreased. Hence, there exists $e\in\mathcal{L}^{C\setminus S}$ such that ${f_e^C}^*(\alpha_1)\leq{f_e^C}^*(\alpha_2)$. Thus, $\mu(\alpha_1)\leq\mu(\alpha_2)$. 
    
    3) By contradiction, suppose  that $\exists l\in\mathcal{L}^S\ |\ {f^S_l}^*(\alpha_1)<{f^S_l}^*(\alpha_2)$.
    Since on all links in $\mathcal{L}^S$ the aggregate load did not increase ($F_l(\alpha_1)\ge F_l(\alpha_2)$), the above implies that ${f^C_l}^*(\alpha_1)>{f^C_l}^*(\alpha_2)$. This implies $\mu(\alpha_1)>\mu(\alpha_2)$, contradicting point 2).\\
    4) 
    Suppose that there $\exists l\in\mathcal{L}^C\,|\,{f_l^{C*}(\alpha_1)}>{f_l^{C*}(\alpha_2)}$. By point 3) we also know that  ${f_l^{S*}(\alpha_1)}\ge{f_l^{S*}(\alpha_2)}$. Hence ${F_l(\alpha_1)}>{F_l(\alpha_2)}$. By \eqref{mon:mu}, this implies $\mu(\alpha_1)>\mu(\alpha_2)$, which contradicts 2). 
\end{proof}
\begin{rem}
    The result above and its proof implicitly assume that  $(\mathcal{L}^S\cap\mathcal{L}^C)\neq\emptyset$. To see that this is always true, assume by contradiction that $(\mathcal{L}^S\cap\mathcal{L}^C)=\emptyset$. Then, it follows 
    \begin{equation*}
    \begin{aligned}
        &\forall l\in\mathcal{L}^S(\alpha),\quad m_l(f_l^*(\alpha))=d_l(F_l^*(\alpha))=\theta(\alpha),\\
        &\forall e\in\mathcal{L}^C(\alpha),\quad m_e(f_e^*(\alpha))>d_e(F_e^*(\alpha))\geq\theta(\alpha),
    \end{aligned}
    \end{equation*}
    which is impossible as vehicles in class $C$ at equilibrium must use links of minimal marginal delay. 
\end{rem}


\begin{prop}
Let Assumptions \ref{as:sm}, \ref{a2} and \ref{a3} hold. Suppose there exists $\alpha_1,\alpha_2\in(0,1)$ such that $\alpha_1<\alpha_2$ and $\mathcal{L}^S(\alpha_1)=\mathcal{L}^S(\alpha_2)$ and $\mathcal{L}^C(\alpha_1)=\mathcal{L}^C(\alpha_2)$. Then, $\mathrm{PoA}(\alpha_1)\geq\mathrm{PoA}(\alpha_2)$.
\label{p5}
\end{prop}

\begin{proof}
    First of all, notice that it suffices to consider only the numerator \eqref{ttt} of $\mathrm{PoA}$, as its denominator is constant. The numerator \eqref{ttt} can be written as follows:
    \begin{equation*}
    \begin{aligned}
        T(f^*(\alpha))
        &=\sum_{l\in\mathcal{L}}{f^S_l}^*(\alpha)\cdot d_l(\alpha)+\sum_{l\in\mathcal{L}}{f^C_l}^*(\alpha)\cdot d_l(\alpha)
    \end{aligned}
    \end{equation*}
    From \ref{i1}) of Proposition \ref{p4},
    \begin{equation}
    \begin{aligned}       &\sum_{l\in\mathcal{L}}{f^S_l}^*(\alpha_2)\cdot d_l(\alpha_2)=\theta(\alpha_2)\sum_{l\in\mathcal{L}}{f^S_l}^*(\alpha_2)\leq\\&\leq\theta(\alpha_1)\sum_{l\in\mathcal{L}}{f^S_l}^*(\alpha_2)=\sum_{l\in\mathcal{L}}{f^S_l}^*(\alpha_2)\cdot d_l(F_l^*(\alpha_1)).
    \end{aligned}
    \label{inter1}
    \end{equation}
    Moreover, because of \ref{i2}) of Proposition \ref{p4}, one can observe that ${f^C}^*(\alpha_1)+({f^S}^*(\alpha_1)-{f^S}^*(\alpha_2))\in\mathcal{F}^C(\alpha_2)$.
    Therefore, if one defines $\varphi:={f^S}^*(\alpha_1)-{f^S}^*(\alpha_2)\ge 0$:
    \begin{equation}
        \begin{aligned}
        &\sum_{l\in\mathcal{L}}{f^C_l}^*(\alpha_2)\cdot d_l(F^*_l(\alpha_2))\leq\sum_{l\in\mathcal{L}}({f_l^C}^*(\alpha_1)+\varphi)\cdot d_l(F_l^*(\alpha_1)),
    \end{aligned}
    \label{inter2}
    \end{equation}
    where the inequality follows from the fact ${f^C_l}^*(\alpha_2)$ minimizes $\sum_{l\in\mathcal{L}}f^C_l\cdot d_l({f^S_l}^*(\alpha_2)+f^C_l)$. The proof is concluded after noticing that summing the inequalities \eqref{inter1} and \eqref{inter2} one gets
    \begin{equation*}
    \begin{aligned}
         T(f^*(\alpha_2))&\leq\sum_{l\in\mathcal{L}}({f_l^S}^*(\alpha_1)+{f_l^C}^*(\alpha_1))\cdot d_l(F_l^*(\alpha_1))=\\&=\sum_{l\in\mathcal{L}}F_l^*(\alpha_1)\cdot d_l(F_l^*(\alpha_1))=T(f^*(\alpha_1)).
    \end{aligned}
    \end{equation*}
\end{proof}

Finally, we study the monotonicity of the PoA under the following additional assumption.

\begin{assum}\label{assump5}
The sets  $\mathcal{L}^S(\alpha)$ and $\mathcal{L}^C(\alpha)$ of active links change a finite number of times as functions of the fleet share $\alpha$ in $[0,1]$.
\end{assum}
\begin{thm}[PoA monotonicity]
Let Assumptions \ref{as:sm}, \ref{a2}, \ref{a3} and \ref{assump5} hold.  $\mathrm{PoA}(\alpha)$ is non-increasing in the fleet share $\alpha$.
    \label{t2}
\end{thm}
\begin{proof}
    Since: i) it follows from Proposition \ref{p2} that the $\mathrm{PoA}$ is Lipschitz continuous, ii) by  Proposition \ref{p5}, the $\mathrm{PoA}$  is 
     non-increasing with $\alpha$ for any interval in which the support does not change,  and iii) by  Assumption \ref{assump5}, the support changes in a finite number of points, one can conclude that the $\mathrm{PoA}$ is non-increasing with $\alpha$  everywhere. 
\end{proof}

\section{Examples}\label{s5}
In the following, we present two examples: the first considers a parallel network as studied in the previous section, while the second explores which results might extend to more general contexts and which may not.



\subsubsection*{Example 1} Consider the example in Figure \ref{f1}. The plots show the evolution of the $\mathrm{PoA}(\alpha)$, the equilibrium loads ${f_l^i}^*(\alpha),\ l=1,2,3,\ i=S,\,C$, and the link delays $d_l(\alpha),\ l=1,2,3$, as functions of $\alpha$, for $\alpha$ varying in $[0,1]$. 
The three plots demonstrate that the $\poa$, the loads of class $S$ and the minimum delay at equilibrium are non-increasing in $\alpha$, while the loads of class $C$ are non-decreasing in $\alpha$. Notice also that, as long as $\alpha\leq\Tilde{\alpha}\approx0.25$, the support of $C$ is included in that of $S$ and $\mathrm{PoA}(\alpha)=\mathrm{PoA}(0)$ for any $\alpha\leq\Tilde{\alpha}$, consistently with Theorem \ref{t1}. Hence this is an example in which a minimum fleet size ($\tilde \alpha$) is needed for affecting the PoA.

\subsubsection*{Example 2} Consider the example in Figure \ref{f2}. The plots depict the behavior of the $\mathrm{PoA}(\alpha)$, the unique equilibrium path flows ${z_p^i}^*(\alpha),\ p=1,2,3,4,\ i=S,\,C$ and the path delays $d_p(\alpha),\ p=1,\dots,4$, as functions of $\alpha$, for $\alpha$ varying in $[0,1]$. 
Although uniqueness of the equilibrium flow $z^*(\alpha)$ is not guaranteed in general, in this case, it follows from the uniqueness of the equilibrium load 
$f^*(\alpha)$ for all $\alpha\in[0,1]$. This is because each path in the network (Figure \ref{f2}) includes a link not shared with any other path, so the load on that link determines the flow on the corresponding path. Since the equilibrium load is unique, so is the equilibrium flow. As in the parallel network case, the $\poa$, equilibrium flows for class $S$, and the minimum path delay are non-increasing in $\alpha$. However, the path flows for class $C$ in this simulation are not necessarily monotone (e.g., the flow on path 4). Whether monotonicity of the $\poa$ can be proven in this more general case remains an open question. 

 \begin{figure}[!t]
\centering
    \begin{subfigure}{0.23\textwidth}
    \includegraphics[width=\textwidth]{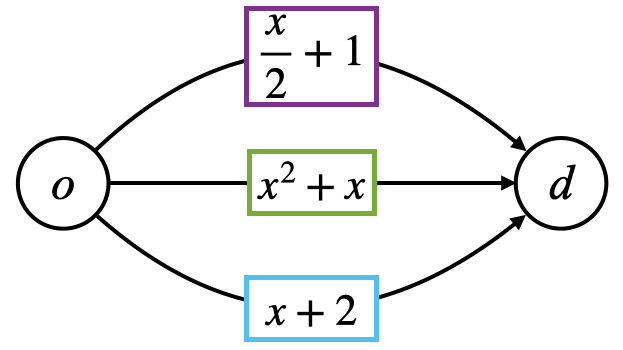}\\
    \vspace{0.5cm}  
    \end{subfigure}
     \begin{subfigure}{0.23\textwidth} 
    \includegraphics[width=\textwidth]{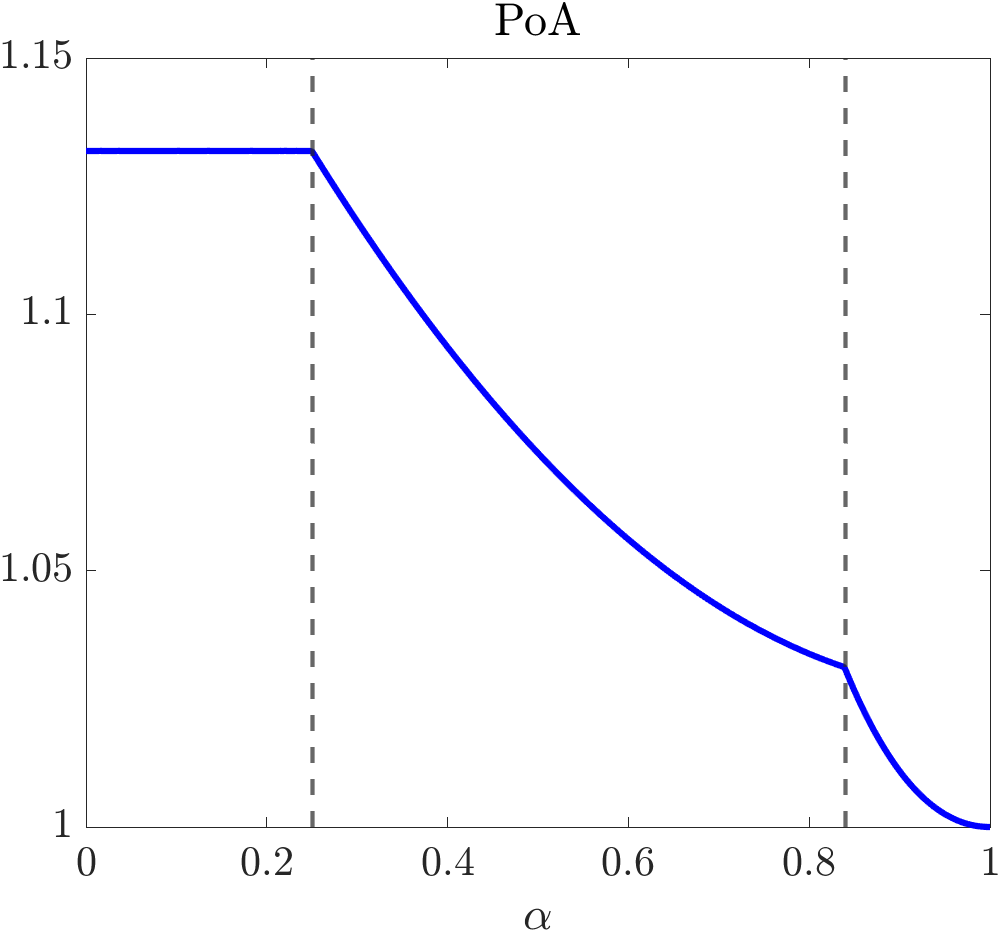}
     \end{subfigure}
    \\ \vspace{0.3cm}\begin{subfigure}{0.23\textwidth}\includegraphics[width=\textwidth]{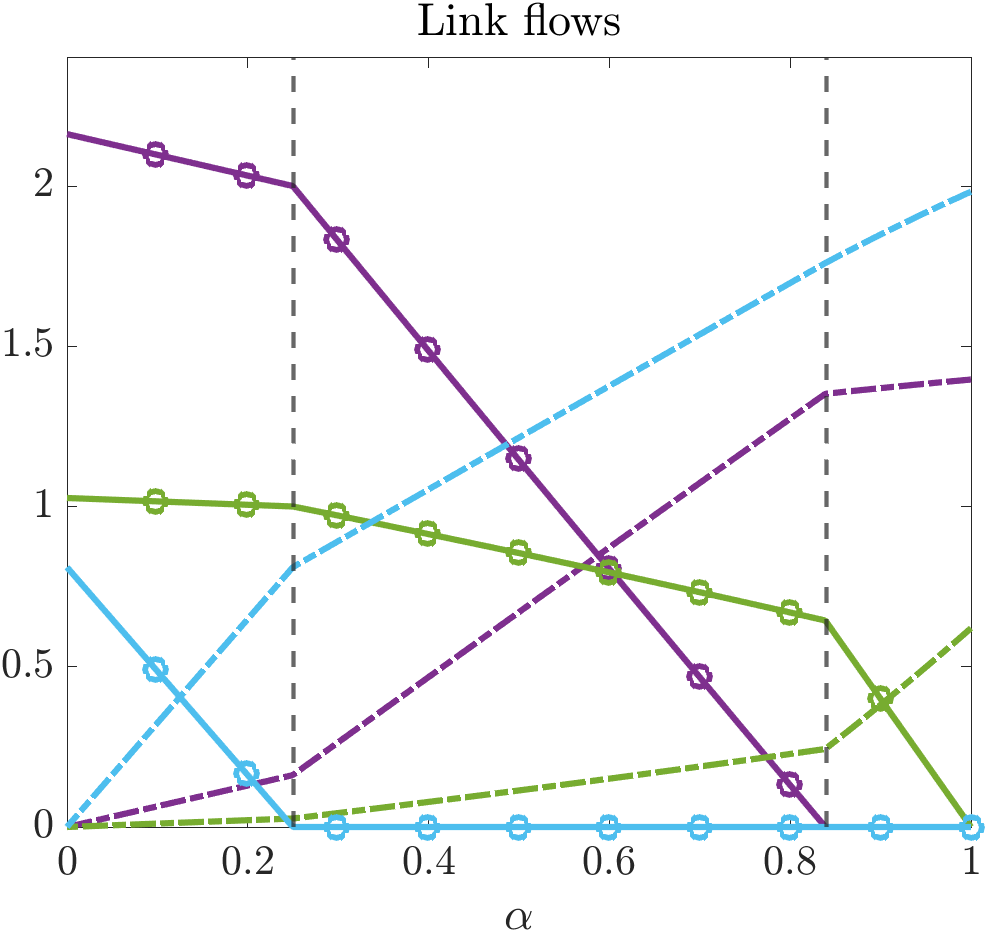} 
     \end{subfigure}
     \begin{subfigure}{0.23\textwidth}
\includegraphics[width=\textwidth]{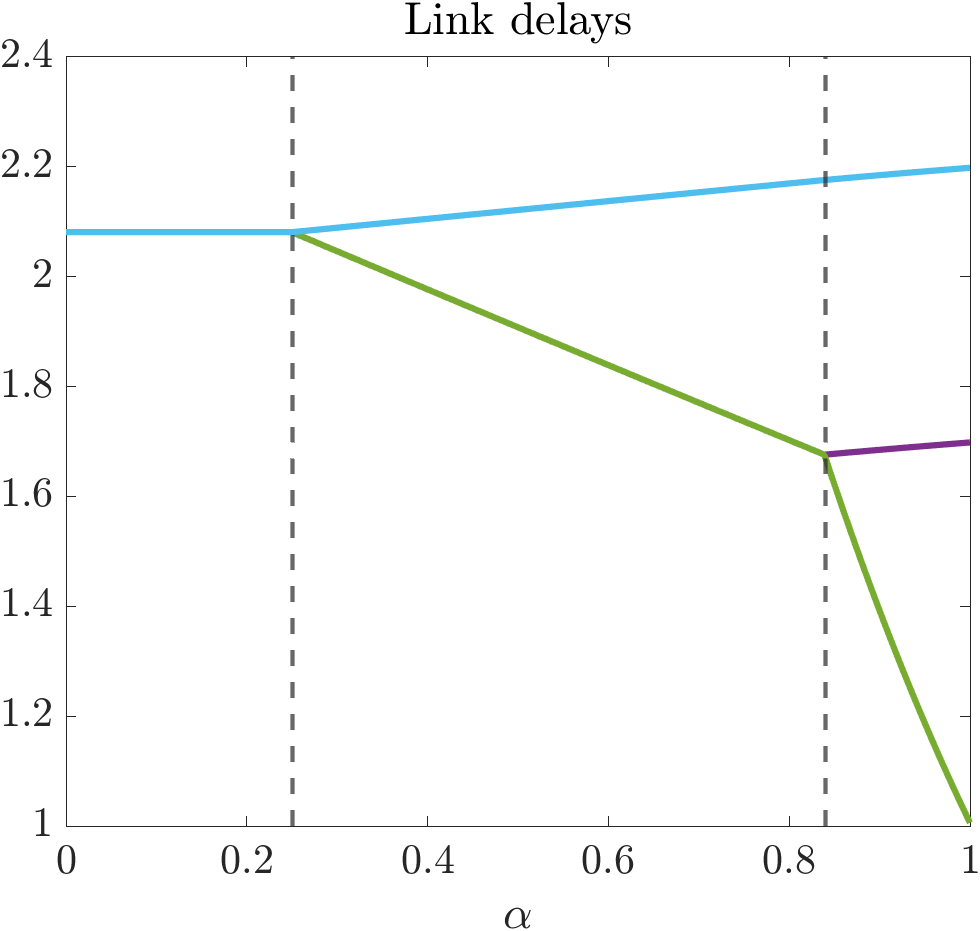}
\end{subfigure}
\caption{A parallel network with link delays labeled. Here, $D=4$. Violet, green, and light-blue lines represent top, middle, and bottom links, respectively. Solid and dashed lines denote loads for classes $S$ and $C$, respectively. Vertical gray dashed lines indicate changes in vehicle class support as $\alpha$ varies.}
\label{f1}
\end{figure}

\begin{figure}[!t]
    \centering
    \begin{subfigure}{0.23\textwidth}
    \includegraphics[width=\textwidth]{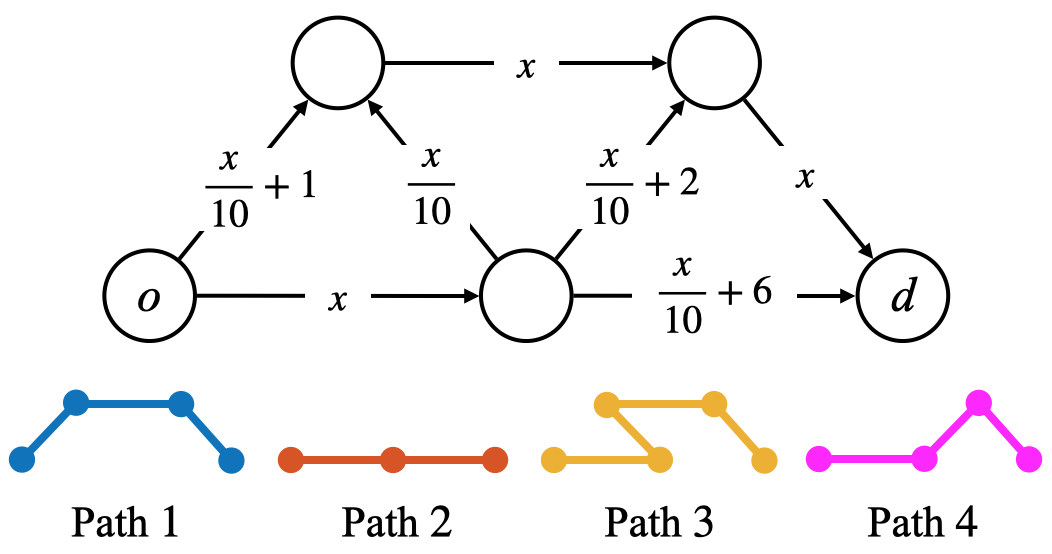}\\
    \vspace{0.5cm}  
    \end{subfigure}
     \begin{subfigure}{0.23\textwidth} 
    \includegraphics[width=\textwidth]{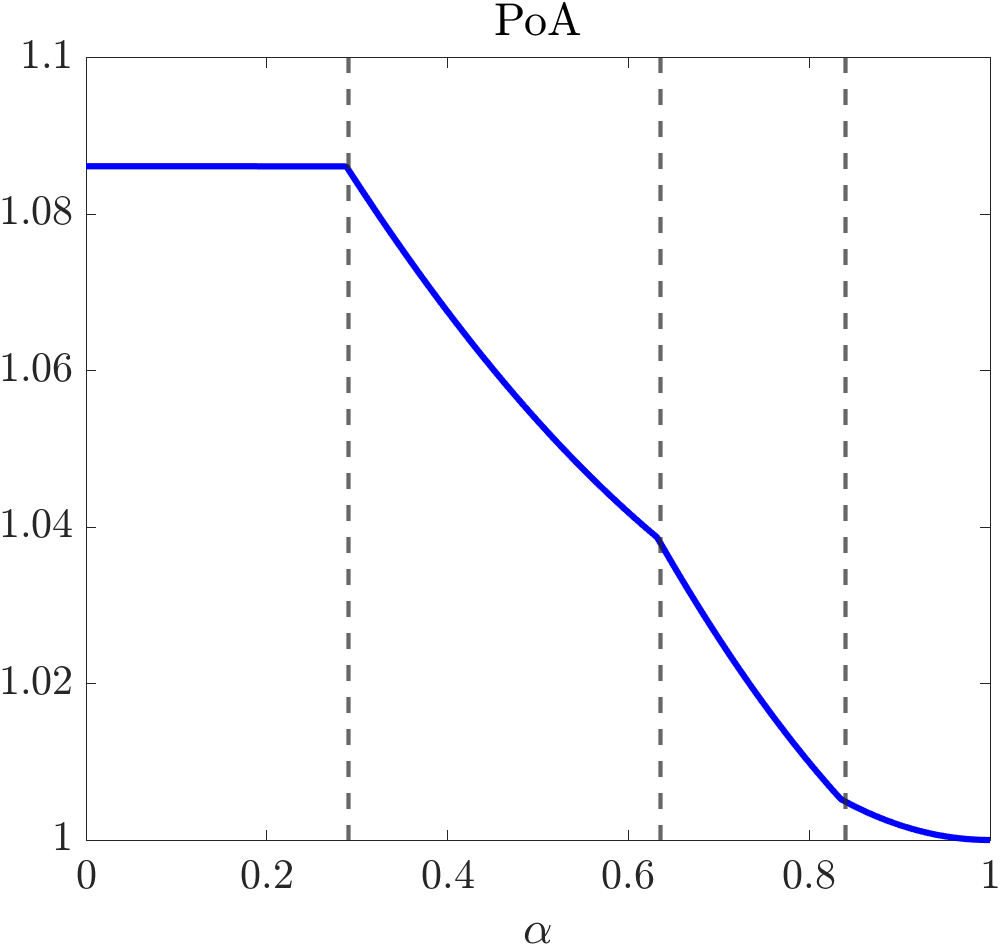}
     \end{subfigure}
    \\\vspace{0,2cm}\begin{subfigure}{0.23\textwidth}\includegraphics[width=\textwidth]{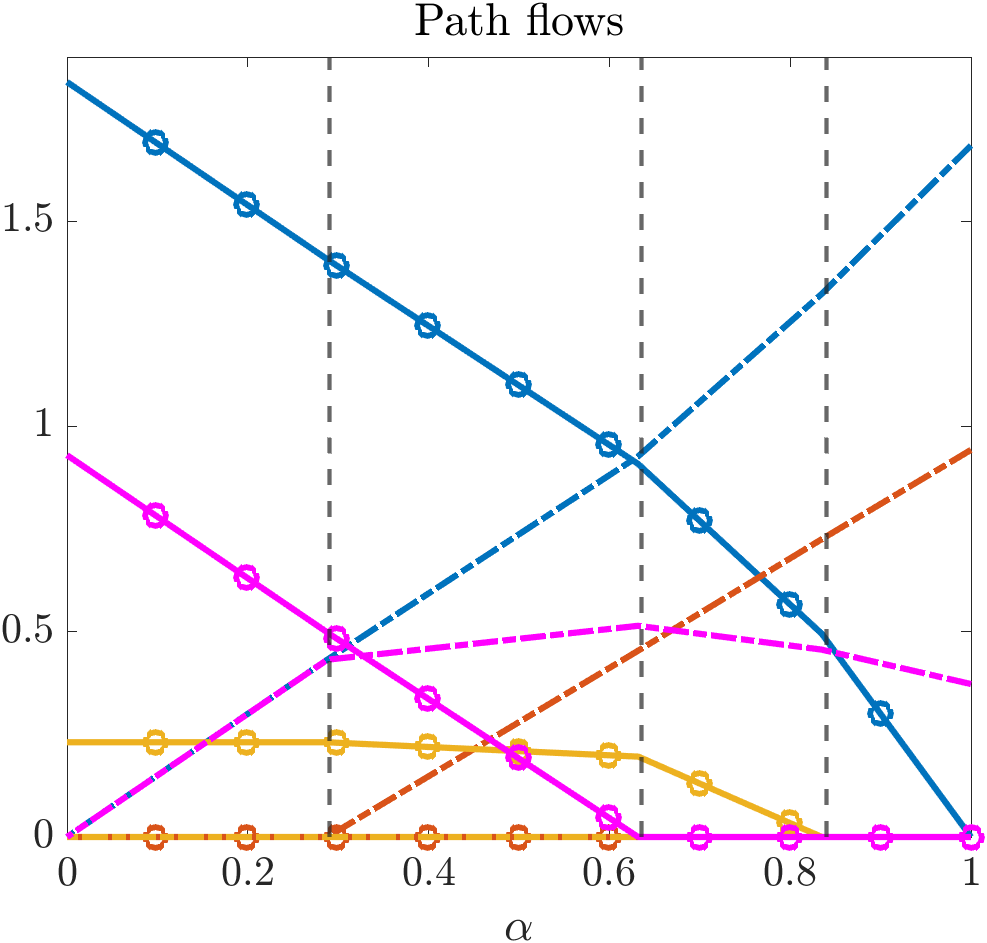} 
     \end{subfigure}
     \begin{subfigure}{0.23\textwidth}
\includegraphics[width=\textwidth]{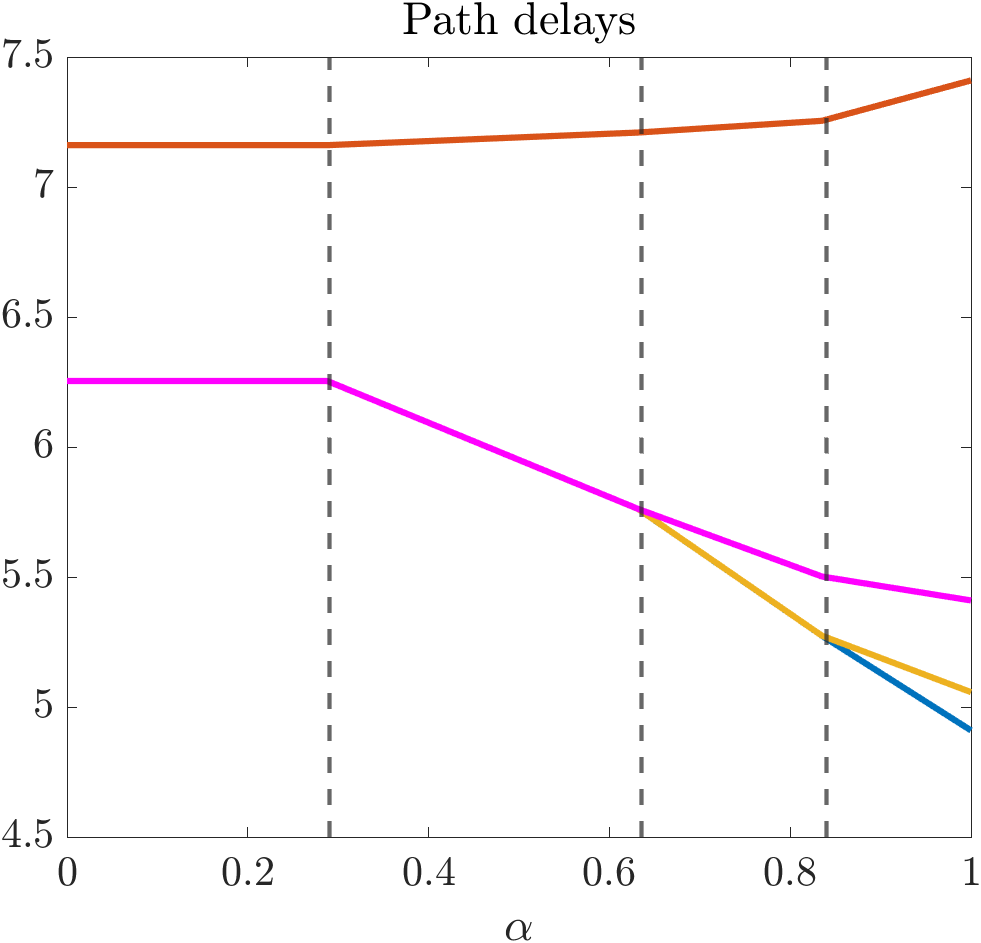}
 \end{subfigure}
    \caption{A network with seven links and four paths. Link delays labeled. Here, $D=3$. Blue, orange, yellow, and magenta lines represent Paths $1$, $2$, $3$, and $4$, respectively. Solid and dashed lines denote flows for classes $S$ and $C$, respectively. Vertical gray dashed lines indicate changes in vehicle class support as $\alpha$ varies.}
    \label{f2}
\end{figure}

\section{Conclusion}\label{s6}

This paper shed lights on how coordinated routing influences the efficiency of transportation networks. We uncover two key phenomena in two-terminal networks. First, we identify scenarios where the coordinated fleet must reach a certain size before impacting the PoA of the unique equilibrium load. Second, we demonstrate that in parallel networks, under suitable assumptions, the PoA weakly decreases as the size of the coordinated fleet increases.

We plan to explore several future directions. First, we aim to more precisely characterize the threshold phenomenon related to the share of the coordinated fleet by establishing clear conditions under which this phenomenon occurs. Additionally, we intend to investigate whether the monotonicity of the PoA holds in more general two-terminal networks, as suggested by the final example in Section \ref{s5}.
Lastly, we aim to extend our analysis to networks with multiple origins and destinations. For this setting, we remark that \cite{batti} already proved that PoA monotonicity does not hold in general. Yet, we believe that establishing sufficient network conditions guaranteeing that the presence of coordinated fleets improves network efficiency would represent an important future contribution. Future work will also be dedicated to investigate the monotonicity of the PoA with more general delay functions, beyond the constraint in equation \eqref{unique}.



\bibliographystyle{IEEEtran}
\bibliography{IEEEabrv,biblio}

\end{document}